\documentclass[11pt, oneside]{article}   	
\usepackage{geometry}                		
\usepackage{xcolor}
\usepackage{authblk}
\usepackage{graphicx}				
\usepackage{amssymb, amsthm, mathtools,bbm}
\usepackage[square,sort,comma,numbers]{natbib}
\usepackage{color}
\usepackage{comment}

\newtheorem{theorem}{Theorem}
\newtheorem{proposition}[theorem]{Proposition}
\newtheorem{lemma}[theorem]{Lemma}
\newtheorem{corollary}{Corollary}
\numberwithin{equation}{section}
\numberwithin{theorem}{section}

\newcommand{\tr}{{\operatorname{Tr}\,}}
\newcommand{\fs}{-}
\newcommand{\cm}[1]{\textcolor{blue}{#1}}

\title{Dynamical Phase Diagram of the REM\\ under independent spin-flips}

\author[1,2]{Chokri Manai}
\author[1,2,3]{Simone Warzel}
\affil[1]{\small Department of Mathematics, TU Munich, Germany}
\affil[2]{\small Munich Center for Quantum Science and Technology, Munich, Germany}
\affil[3]{\small Department of Physics, TU Munich, Germany}
\date{\today}							

\begin{document}
\maketitle
\begin{abstract}
	We study the energy landscape of the Random Energy model (REM)  integrated along trajectories of the simple random walk on the hypercube.  We show that the quenched cumulant generating function of the time integral of the REM energy undergoes phase transitions in the large $N$ limit for trajectories of any time extent, and identify phases distinguished by the activity and value of the time integral. This is achieved by relating the dynamical behavior to the spectral properties of Hamiltonians associated with the Quantum Random Energy Model (QREM). Of independent interest are deterministic $ \ell^p $-properties of the resolvents of such Hamiltonians, which we establish. 
\end{abstract}


\bigskip

\section{Introduction}
Dynamical expressions of glassy behavior are a central topic in the non-equilibrium theory of spin glass models. Much of the mathematical work concerns one of the versions of Glauber or trap dynamics in the simplest spin glass, namely the random energy model (REM) and its relatives, see e.g.\ \cite{Arous:2003lr,Arous:2003nt,Arous:2008,Cerny:2017ie,Gayrard:2019yr} and references therein. These  dynamics are Markov processes that interact in the sense that the jump rates depend on the configuration and its neighbors' energy value. 
The aforementioned works are devoted to the study of typical trajectories in this dynamics, such as their ageing.

In the present work, by contrast, we study the atypical behavior of the REM's energy for rare trajectories of the infinite-temperature limit of the Glauber dynamics, namely the simple random walk. In the last decades, there have been grown interest in the sampling and study of rare events in physical systems due to their importance, among other aspects, in the behavior of models outside of equilibrium \cite{Bolhuis2002, Derrida2007,Jack2015, Touchette2009}. A closed-form expression for the large-deviation function is derived, which reveals a dynamical phase diagram consisting of three phases distinguished by the activity of the jump process and the atypical, time-integrated value of the REM's energy.  
Such trajectory-phase transitions  are known to reveal dynamical phases which are generally distinct from their static counterparts \cite{Lecomte2007, Merolle2005, Garrahan2018, Jack2020, Garrahan2007, Nyawo2016,  Hedges2009, Speck2012, Vasiloiu2020}.  This is also the case in our study, which is a follow-up of the investigation of the trajectory phase-diagram of the REM under the simpler all-to-all dynamics~\cite{GMW23}.

Although our point of view is different from that taken in the above-mentioned studies of Glauber dynamics, it is natural to conjecture that the atypical behavior of the REM seen under the simple random walk is what is typically observed in a trap dynamics at very high temperature \cite{BBC18,BC22}.

\subsection{Simple random walk on the Hamming cube}
The simple random walk $ \pmb{\omega}: [0,\infty) \to \mathcal{Q}_N $ is the most elementary stochastic dynamics on the configuration space $ \mathcal{Q}_N  := \{ -1,1 \}^N  $ of $ N $ Ising spins. 
As a continuous-time Markov jump process, it is uniquely characterized by the rate
\begin{equation}
w(\pmb{\sigma}\to \pmb{\tau}) = \begin{cases} 1, & d(\pmb{\sigma}, \pmb{\tau}) = 1, \\
					0,  &  \mbox{else}, 
					\end{cases}
\end{equation}
for jumping from a spin configuration $ \pmb{\sigma} = (\sigma_1, \dots , \sigma_N) \in \mathcal{Q}_N $ to any neighboring configuration $ \pmb{\tau} $, that is, a configuration at unit Hamming distance $ d(\pmb{\sigma}, \pmb{\tau})  = \frac12 \sum_{i=1}^{N} |\sigma_i - \tau_i| $. The random waiting times $ t \geq 0 $ for a  jump from $   \pmb{\sigma} $ to any neighbor are distributed according to the exponential law $ r( \pmb{\sigma}) \exp\left(- r(\pmb{\sigma})t\right) $. 
In our units, the escape rates $ r(\pmb{\sigma})  := \sum_{ \pmb{\tau} \neq \pmb{\sigma}}  w(\pmb{\sigma}\to \pmb{\tau})  = N $ of the simple random walk are configuration-independent and equal to the particle number $N$.  This corresponds to a unit flip rate of the individual spins $\sigma_j$. Accordingly, the generator $ W $ of this Markov process is the Laplacian on the Hamming cube: 
\begin{equation}
	( Wf)(\pmb{\sigma})  := \sum_{
		 \pmb{\tau} \in  \mathcal{Q}_N }  
		w(\pmb{\sigma}\to \pmb{\tau}) \,  \left( f(\pmb{\tau}) - f(\pmb{\sigma}) \right) = ( Tf)(\pmb{\sigma})   - N f(\pmb{\sigma})  . 
 \end{equation}
The linear operator $T$ acting on observables $ f : \mathcal{Q}_N \to \mathbb{R} $ stands for the adjacency matrix of the Hamming cube. The adjoint $ W^* $ of the Markov generator, which in the present set-up agrees with $ W $, acts on probability distributions $ p_t (\pmb{\sigma})  $, an the master equation 
\begin{equation}
\partial_t   p_t (\pmb{\sigma}) = \left( W^*p_t \right)(\pmb{\sigma})  
\end{equation}
governs the time evolution of any initial distribution $ p_0(\pmb{\sigma})$.

\subsection{Trajectory observables and their large deviations}
Given an energy function $ U: \mathcal{Q}_N  \to \mathbb{R} $, one may explore this  function along any trajectory $ \pmb{\omega} : [0, \infty ) \to \mathcal{Q}_N $ of the random walk: 
$$
U_t[ \pmb{\omega}] := \int_0^t U\left( \pmb{\omega}(s)\right) \ ds , \quad t > 0 . 
$$
We will be interested in the probability distribution of $ U_t $ under the law $ P_t $ of the simple random walk up to time $ t $ with the initial spin configurations equally distributed, that is, under the unique stationary distribution $ p_\infty(\pmb{\sigma}) = 2^{-N} $ of the simple random walk. The main result of this note is a proof of a large deviation principle for this distribution in the limit of large system size $ N $ for trajectories of any time extent $t$. 
The large deviation principle  is  described in terms of the moment generating function 
\begin{equation}\label{eq:mom1}
	\int e^{-\lambda U_t[ \pmb{\omega}] } \ P_t(d\pmb{\omega}) = \sum_{\pmb{\sigma} \in \mathcal{Q}_N} \left( e^{t \left( W^* - \lambda U \right)} p_\infty\right)(\pmb{\sigma}) = 
	\frac{1}{2^N}  \sum_{\pmb{\sigma}, \pmb{\tau}\in \mathcal{Q}_N}  \langle  \pmb{\sigma} | e^{t \left( T - N\mathbbm{1}  - \lambda U  \right)} | \pmb{\tau} \rangle   . 
\end{equation}
The first equality follows from the Feynman-Kac formula (cf.~\cite{KLW21, Leschke:2021xw}), and the second equality is by definition of the equilibrium distribution $p_\infty$. 
The identity \eqref{eq:mom1} involves  the exponential of the tilted generator 
$ W^* - \lambda U $ with $ \lambda \in \mathbb{R} $, which may also be viewed as a self-adjoint matrix 
acting on the tensor product Hilbert space $ \mathcal{H} := \otimes_{j=1}^N \mathbb{C} ^2$. In the last step of~\eqref{eq:mom1} and subsequently, we use Dirac's notation for the canonical orthonormal tensor product basis,
$ | \pmb{\sigma}\rangle  = \otimes_{j=1}^N  |\ \sigma_j\rangle$.  This orthonormal basis is the eigenbasis of $ U $, that is, $ U |\pmb{\sigma}\rangle = U(\pmb{\sigma})  |\pmb{\sigma}\rangle $.  Introducing the normalized flat vector
\begin{equation}\label{eq:flat}
 | \fs \rangle :=  \frac{1}{\sqrt{2^N}} \sum_{\pmb{\sigma}\in \mathcal{Q}_N} | \pmb{\sigma} \rangle \in  \mathcal{H} ,
\end{equation}
one arrives at the expression
$ \langle \fs |  e^{t \left( W - \lambda U \right)  } | \fs \rangle $ for the right side of~\eqref{eq:mom1}.  
This connects the question concerning the atypical behavior of $ U_t$  under the law $ P_t $ to a spectral problem, namely the properties of the semigroup on $ \mathcal{H} $ generated by $ T - N\mathbbm{1}  - \lambda U $ in the flat vector.   In this context, it is useful to note that the adjacency operator $ T= \sum_{j=1}^N X_j $ agrees with the sum of Pauli-$ X $ matrices, 
 which flip the $ j $th spin, that is, $  X_j |  \pmb{\sigma}  \rangle = | \sigma_1, \dots , - \sigma_j , \dots , \sigma_N \rangle $.   

To measure the activity of the jump process, it is convenient to introduce yet another tilting in the generator, 
\begin{equation}
W_{\lambda, s} := e^{-s} \ T - N \mathbbm{1} + \lambda U , \quad s, \lambda \in \mathbb{R} ,
\end{equation}
which modifies the individual spin-flip rate to $ e^{-s} $. but keeps the jump rate constant at $ N $. 
The associated moment-generating function is
\begin{equation}
	Z(t,\lambda,s)  := \langle \fs |  e^{t W_{\lambda,s}  } | \fs \rangle .
\end{equation}
What plays the role of a free energy for the trajectories is the scaled cumulant generating function (SCGF) given by
$$
\theta_N(t,\lambda,s) := \frac{1}{N t} \ln Z(t,\lambda,s) . 
$$
In case the limit $\theta(t,\lambda,s) :=  \lim_{N\to \infty} \theta_N(t,\lambda,s) $ exists, 
the G\"artner-Ellis theorem~\cite[Thm. 2.3.6]{DemZeit98} implies that  the Legendre-Fenchel transformation 
\begin{equation}
\varphi(t,u,s)  := \sup_{\lambda } \left( u \lambda -  \theta(t,\lambda,s) \right)  
\end{equation}
governs the  large deviations of $ U_t $, that is, for any Borel set $ I \subset \mathbb{R} $ and any $ t > 0 $:
\begin{align}\label{eq:LDPU}
 - \inf_{u \in I^\circ } t \varphi(t, u,0) & \leq    \liminf_{N\to \infty} \frac{1}{N} \ln P_t\left((Nt)^{-1} U_t \in I \right) \notag \\ 
&  \leq  \limsup_{N\to \infty} \frac{1}{N} \ln P_t\left( (N t)^{-1} U_t \in I  \right) = - \inf_{u \in \overline{I}} t \varphi(t, u,0) . 
\end{align}

Thanks to convexity, the partial derivative $ - \partial_s  \theta(t,\lambda,0) := \partial_s  \theta(t,\lambda,s) \big|_{s=0} $, whenever it exists,  agrees up to a sign with the average activity per unit space and time. More precisely, the asymptotic average of the number of configuration changes in a trajectory is asymptotically given by
\begin{equation}
- \partial_s  \theta(t,\lambda,0)   = \lim_{N\to \infty}  \frac{1 }{N  } \sum_{\pmb{\sigma}\neq \pmb{\tau} } \int_0^t  P_{t,\lambda}\left( \pmb{\omega}(s+0)=  \pmb{\sigma}\; \mbox{and} \; \pmb{\omega}(s-0)=  \pmb{\tau} \right)   \ \frac{ds}{t} , 
\end{equation}
with $  P_{t,\lambda} $ the tilted probability measure corresponding to $  e^{-\lambda U_t[\pmb{\omega}]} \ P_{t,\lambda}(d\pmb{\omega}) / Z(t,\lambda,0) $. 

\subsection{Trajectory phase transition for the REM} 
Our main result is the existence of the limit $ N \to \infty $ and an explicit expression for the SCGF in case $ U $ is the Random Energy Model (REM) \cite{Derrida:1980mg,Bov06}.
The REM,
$U: \mathcal{Q}_N \to \mathbb{R}$, is a Gaussian random field with randomness independent of the Markov process, in which the values $ U(\pmb{\sigma}) $ are distributed independently for all $ \pmb{\sigma} \in \mathcal{Q}_N $ with identical normal law uniquely characterized by zero mean and variance $ N $, that is, denoting the law by $ \mathbb{P} $ and the corresponding expectation value by $ \mathbb{E} $, we have
$$
\mathbb{E}\left[U(\pmb{\sigma})\right] = 0 , \quad \mathbb{E}\left[U(\pmb{\sigma}) U(\pmb{\tau})\right] = \begin{cases} N , & \pmb{\sigma} = \pmb{\tau} , \\ 0 , & \mbox{else.}\end{cases}
$$ 
The units are chosen so that the REM's large deviations occur on order $ N $ which agrees with the norm of $ T $. Up to a shift, the semigroup generator $ e^{-s} T - N \mathbbm{1} - \lambda U $ coincides with the Hamiltonian of the Quantum Random Energy Model (QREM) -- one of the simplest quantum spin glass models~\cite{Goldschmidt:1990kr,Manai:2020ta,Manai:2021nu,MaWa20}.\\

The proof of the following main result, which can be found in Section~\ref{sec:proof},  builds on the comprehensive spectral analysis of the QREM \cite{Manai:2020ta, Manai:2021nu}, but requires substantial new technical insights. 

\begin{theorem}\label{thm:mainasym}
For any $ t  > 0, \lambda \geq  0 $, $ s \in \mathbb{R} $ the scaled cumulant generating function for the REM converges $ \mathbb{P} $-almost surely with  limit given by
\begin{equation}\label{eq:mainasym}
\lim_{N \to \infty} \theta_N(t,\lambda,s) = \theta(t,\lambda,s) \coloneqq \max\left\{ e^{-s}   , t^{-1} p_{\mathrm{REM}}(t \lambda )  \right\}  -1 ,
\end{equation}
with 
\begin{equation}\label{eq:defthetanull}
	p_{\mathrm{REM}}( \beta ) := \begin{cases} 
		\frac{\beta^2}{2} , & \beta \leq  \beta_c := \sqrt{2 \ln 2}, \\
		\beta \beta_c - \ln 2 , & \beta > \beta_c .
		\end{cases}
\end{equation}
\end{theorem}
Some remarks are in order. 

By the symmetry of the REM's distribution, we restrict ourselves to the case $ \lambda \geq 0 $ without loss of generality. 

The quantity defined in~\eqref{eq:defthetanull} is the pressure corresponding to the REM's static (normalized) partition function at inverse temperature $ \beta $:
\begin{equation}\label{eq:REMp}
p_{\mathrm{REM}}( \beta )  = \lim_{N\to \infty} \frac{1}{N} \ln \frac{1}{2^N} \sum_{\pmb{\sigma}} e^{-\beta U(\pmb{\sigma}) }  .
\end{equation}
The critical value $ \beta_c = \sqrt{2 \ln 2} $ is the inverse of the REM's freezing temperature into a spin glass phase with 1-step replica symmetry breaking, cf.~\cite{Derrida:1980mg,Bov06}.  

The limit of the SCGF $\theta(t,\lambda,0) $ coincides with the limit in the case of the all-to-all dynamics, for which the generator is 
$ W = N ( | \fs \rangle \langle \fs | - \mathbbm{1} ) $. This model was studied in~\cite{GMW23} (see also~\cite{ASW15}), and the Legendre-Fenchel transform of the SCGF was computed:
$$
\varphi(t,u,0)  := \sup_{\lambda } \left( u \lambda -  \theta(t,\lambda,0) \right)  =
 \begin{cases} |u| \sqrt{\frac{2}{t}}, & |u| \leq \min\left\{ \sqrt{2t} , \beta_c \right\},  \\  1+ \frac{u^2}{2t}, & \mbox{else}, \\ \infty,  & |u| > \beta_c . \end{cases}
$$
\begin{figure*}[!h]
    \centering
    \begin{minipage}{0.49\textwidth}
    \includegraphics[width=\textwidth]{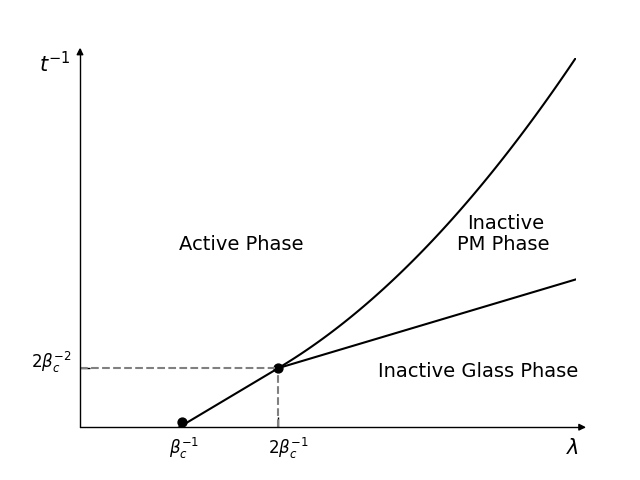}
    \end{minipage}
    \begin{minipage}{0.49\textwidth}
    \includegraphics[width=\textwidth]{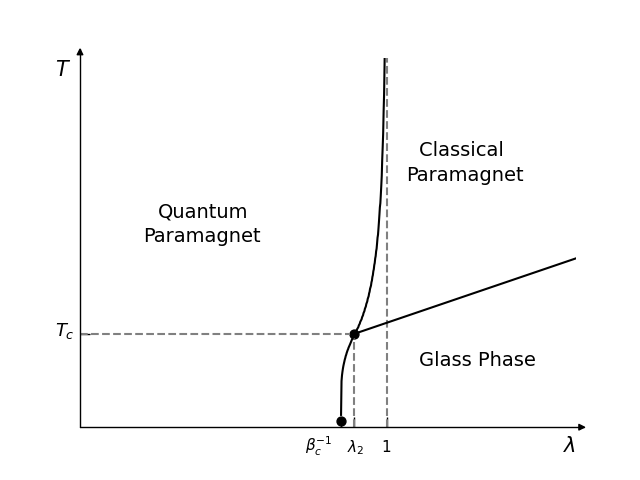}
    \end{minipage}
	\caption{\small Comparison of the trajectory REM phase diagram and the QREM static phase diagram \\[1ex]
        1. The left depicts the dynamical phase diagram as a function of the potential strength $\lambda$ and the inverse of the trajectory length $t^{-1}$. The full lines indicate phase transitions between the Active Phase and the two Inactive Phases. For $\lambda \leq \lambda_1 \coloneqq \beta_c^{-1}$ the system is in the Active Phase for all times, whereas for $\lambda_1 < \lambda \leq \lambda_{2} \coloneqq 2 \beta_c^{-1}$ the model undergoes a phase transition for large times to the Inactive Glass Phase. For $\lambda > \lambda_2$ all three phases occur at different trajectory lengths. The unique triple point is located at $\lambda = \lambda_2$ and $t_c = \beta_c^2/2$. \\
        2. The right shows the static QREM phase diagram in terms of the potential strength $\lambda$ and the temperature $T = \beta^{-1}$ at $ s = 0 $. The QREM phase diagram consists of three phases: a quantum paramagnetic (QP) phase, a classical glass phase, and a classical paramagnetic phase. For $\lambda \leq \lambda_1 \coloneqq \beta_c^{-1}$, the model is for all temperatures $T \geq 0$ in its quantum paramagnetic phase. If $\lambda_1 < \lambda \leq \lambda_2 \coloneqq  \frac{\beta_c}{\mathrm{arcosh}(2)}$, the system undergoes a transition from the QP phase to its glass phase while the temperature is lowered; and for $\lambda_2 < \lambda < \lambda_3 \coloneqq  1$ there are two phase transitions: first from the QP phase to the classical paramagnetic phase and then to the glass phase for decreasing temperatures. In contrast to the active phase, the QP phase disappears for $\lambda \geq \lambda_3$. The unique triple point is located at $\lambda = \lambda_2$ and $T = T_c \coloneqq \mathrm{arcosh}(2)$.
   } 
	\label{fig:phase}
\end{figure*}
Theorem~\ref{thm:mainasym} hence yields the same dynamical phase diagram as the all-to-all dynamics~\cite{GMW23}. It is made up of three regimes depicted in~Figure~\ref{fig:phase}, which we briefly summarize:
\begin{enumerate}
\item  An  {\em active} dynamical phase in which the Markov generator $ W $ dominates over the tilting, and which is characterized by $ \theta(t,\lambda,0) = 0 $ and the specific activity being unity, 
$
- \partial_s  \theta(t,\lambda,0) = 1 $.
It is separated from the remaining regimes by a first-order transition line.
This regime persists for all $ | \lambda | < \beta_c (2t)^{-1} + \beta_c^{-1} $ in case $ t^{-1} < 2 \beta_c^{-2} $ and $ |\lambda | <  \sqrt{2 t^{-1}} $ in case $ t^{-1} \geq  2 \beta_c^{-2} $. 
\item A regime of vanishing activity, $
- \partial_s  \theta(t,\lambda,0) = 0 $, 
which occurs for $ t^{-1}  <2 |\lambda| \beta_c^{-1}  - 2  \beta_c^{-2}  $ and which is dominated by the REM's extreme values where the system localizes. This regime is related to the spin-glass phase of the REM and 
we call this the {\em Inactive Glass} dynamical phase.

\item The remaining parameter regime corresponds to a second inactive regime which we term {\em Inactive Paramagnetic} dynamical phase. 
It occurs only if $ t^{-1} > 2/\beta_c^2 $ and is related to the classical paramagnetic phase of the REM.
 \end{enumerate}

Remarkably, the large deviations of the trajectory observable $ U_t $ are unchanged compared to the all-to-all dynamics for the simple random walk, where several visits of the same site are overwhelmingly more likely. In that sense, Theorem~\ref{thm:mainasym} is another striking consequence of the roughness of the REM potential, which causes the operator $ e^{-s} T - \lambda U $ to almost decompose direct into a sum of  its potential part and the transversal field $T$. 
Compared to the all-to-all dynamics, differences can be found on the level of the free energy.  
For the QREM it is given for all $ \beta , \lambda \geq  0 $, $ s \in \mathbb{R} $ by~\cite{Manai:2020ta}
 \begin{equation}\label{eq:pQREM}
 \lim_{N\to \infty} \frac{1}{N} \ln \frac{1}{2^N} \tr e^{ \beta (e^{-s} T - \lambda U) } 
 = \max\left\{\ln\cosh(\beta e^{-s}), p_{\textrm{REM}}(\beta\lambda) \right\} .
 \end{equation}
 In the all-to all dynamics, the logarithm of the hyperbolic cosine in the right side is changed to $\beta e^{-s} - \ln 2 $, cf.~\cite[Thm 5.1]{GMW23}. \\

The QREM can be naturally seen as the limit of the class of $p$-spin models. In fact, the free energy is known  \cite{MW20b, KMW25a} to converge to that of QREM as $ p \to \infty $. In addition, aging in the trap dynamics of $ p $-spin glasses belongs to the REM universality class~\cite{Arous:2008}. It is therefore natural to expect that a similar statement holds also for the SCGF, namely that the limit $ N \to \infty $ for a $ p $-spin glass and the subsequent limit $ p \to \infty $ converges to the right side in~\eqref{eq:mainasym}. As a lower bound, Lemma~\ref{lem:lb} below shows that this applies.

Arguably, the most studied spin glass in a transversal field is the Quantum Sherrington-Kirkpatrick model with $ p = 2 $ , where recently an infinite-dimensional Parisi formula has been established \cite{MW24}. Although we do not expect to gather insight into a closed-form expression for the trajectory phase diagram of the
Sherrington-Kirkpatrick model, it might be interesting to see whether the trajectory phase diagram in case $ p = 2 $ also only has two phases as its free energy analogue~\cite{LMRW21,Young17}.

\section{Proof of the main result}\label{sec:proof}

\subsection{Lower bound}
The proof of Theorem~\ref{thm:mainasym} is based on asymptotically coinciding upper and lower bounds. The latter is summarized in the following. Its proof is essentially contained in~\cite{GMW23}. 

\begin{lemma}\label{lem:lb}
    Let $ U :\mathcal{Q}_N \to \mathbb{R} $ be such that 
    \begin{equation}\label{eq:lemlb} \lim_{N\to \infty} \frac{1}{N 2^N} \sum_{\pmb{\sigma} \in \mathcal{Q}_N} U(\pmb{\sigma}) = 0 .
    \end{equation}
    Then for any $ t> 0 $, $ s,\lambda \in \mathbb{R} $:
    \begin{equation}
    \liminf_{N\to \infty} \frac{1}{tN} \ln \langle \fs | e^{t W_{\lambda,s}} | \fs \rangle  \geq \max\left\{ e^{-s}   , t^{-1} p_{U}(t \lambda )  \right\}  -1 
    \end{equation}
    where $\  \displaystyle  p_U(\beta) \coloneqq\liminf_{N\to \infty} \frac{1}{N} \ln \frac{1}{2^N} \sum_{\pmb{\sigma} \in \mathcal{Q}_N} e^{\beta U(\pmb{\sigma}) } $ with $ \beta \in \mathbb{R} $.
\end{lemma}
\begin{proof}
 We use Jensen's inequality and the eigenvalue equation $ T |\fs\rangle = N | \fs \rangle $ to conclude
$$
 \ln \langle \fs | e^{t W_{\lambda,s}} | \fs \rangle \geq t  \langle \fs | W_{\lambda,s} | \fs \rangle  = N t \ (e^{-s} -1 ) +   \frac{t \lambda}{2^N} \sum_{\pmb{\sigma}} U(\pmb{\sigma}) .
$$
By assumption, the lower limit of the sum on the right side converges to zero, when deviding by $ t N $. For another lower bound, which is sharper in case $ t  e^{-s}   <  p_{U}(t \lambda )   $, we estimate using the non-negativity of the matrix elements of the semigroup:
\begin{align*}
\langle \fs | e^{t W_{\lambda,s}} | \fs \rangle = \frac{1}{2^N} \sum_{\pmb{\sigma}, \pmb{\tau}}  \langle \pmb{\tau}  | e^{t W_{\lambda,s}} | \pmb{\sigma} \rangle & \geq  \frac{1}{2^N} \sum_{\pmb{\sigma}}  \langle \pmb{\sigma}  | e^{t W_{\lambda,s}} | \pmb{\sigma} \rangle \\
& \geq  \frac{1}{2^N} \sum_{\pmb{\sigma}} \exp( \langle \pmb{\sigma} | \ t W_{\lambda,s} | \pmb{\sigma} \rangle) = e^{-tN}  \frac{1}{2^N} \sum_{\pmb{\sigma}} e^{t\lambda U(\pmb{\sigma})} 
\end{align*}
The second bound is again by Jensen's inequality.  This yields the alternative lower bound. 
\end{proof}

\subsection{Large deviations of the density of states measure}
While the lower bound stems from~\cite{GMW23}, the upper bound requires a new idea. It is based on controlling the spectral density of the high-energy QREM eigenfunctions in the flat state. More precisely, we fix $ \lambda , \Gamma > 0 $ and $ \delta > 0 $ and  set 
\begin{equation}\label{eq:proj}
Q_{\delta} := 1_{(\delta N,\infty)}\left[  \Gamma T- \lambda U  \right],  
\end{equation}
the spectral projection of the linear operator $  \Gamma  T- \lambda U  $ associated to the open energy interval $ (\delta N,\infty) $.  
By the known relation of the asymptotic behavior of the Laplace transform and the Legendre-Fenchel transform, the main result, Theorem~\ref{thm:mainasym}, can be recast as an asymptotic result on the density of states measure in the flat states, that is, $  \langle \fs | Q_{\delta \lambda } | \fs \rangle $ with varying $ \delta \lambda $.  In our proof, the main stumbling block is to control this quantity in case $  \delta \lambda > \Gamma $, which is the content of
\begin{proposition}\label{lem:qlargedev}
For any $ \lambda \delta >\Gamma \geq  0, $ and $ \mathbb{P} $-almost all realizations of the REM:
\begin{equation}\label{eq:qlargedev}
	\limsup_{N\to \infty} \frac{1}{N} \ln \langle \fs | Q_{\delta \lambda } | \fs \rangle  \leq   - \begin{cases} \frac{\delta^2}{2} , & \delta \leq \beta_c, \\
		\infty , & \delta < \beta_c . 
		\end{cases}   
\end{equation}
\end{proposition}


Three remarks are in order:
\begin{enumerate}
\item 
In case $ \Gamma = 0 $, the above statement reduces to the basic large-deviation result concerning the equality of the empirical distribution and the probability:
\begin{equation}\label{eq:clLD}
	\lim_{N\to \infty} \frac{1}{N} \ln  \frac{| \mathcal{L}_\delta| }{2^N}  =   - \begin{cases} \frac{\delta^2}{2} , & \delta < \beta_c, \\
		\infty , & \delta \geq \beta_c . 
		\end{cases}
\end{equation}
where $ | (\cdot) | $ denotes the volume of the set and 
\begin{equation}\label{eq:ldset}
	\mathcal{L}_\delta := \left\{ \pmb{\sigma} \in \mathcal{Q}_N \ | \ U(\pmb{\sigma}) < - \delta N \right\} 
\end{equation}
denotes the set of extreme negative deviations of the REM. The case $ \delta <\beta_c $ follows from the Legendre-Fenchel duality 
$$
 \sup_{\beta> 0 } \left( \beta \delta - p_{\mathrm{REM}}(\beta) \right) =  \begin{cases} \frac{\delta^2}{2} , & \delta \leq \beta_c, \\
		\infty , & \delta > \beta_c . 
		\end{cases}
$$
The case $ \delta \geq \beta_c $ is consistent with the well-known asymptotics of the REM's minimum energy
\begin{equation}\label{eq:REMmin}
    \min_{\pmb{\sigma}\in \mathcal{Q}_N} U(\pmb{\sigma}) = - \beta_cN + \frac{\ln(N \ln 2)}{2 \beta_c} + \mathcal{O}(1),
\end{equation} 
 see e.g. \cite{Bov06}.
\item Replacing the flat state by the tracial state in Proposition~\ref{lem:qlargedev}, previous works on the spectral properties of the QREM~\cite{Manai:2020ta,MaWa22} yield
\begin{align}\label{eq:traceLD}
	\lim_{N\to \infty} \frac{1}{N} \ln 2^{-N } \tr Q_{\delta \lambda }  & = - \begin{cases} \frac{\delta^2}{2} , & \delta < \beta_c \\
		\infty , & \delta \geq \beta_c \end{cases}  
\end{align}
for $ \lambda \delta >\Gamma \geq  0. $ That is, $\tr Q_{\delta \lambda }$ and the number of classical large deviations $| \mathcal{L}_\delta|$ agree on an exponential scale in this regime.  
We have included a short proof in the Appendix. 
\item In view of \eqref{eq:traceLD}, it is tempting to conjecture that the bound in \eqref{eq:qlargedev} is in fact sharp for $ \delta \neq \beta_c$, that is, one has the following asymptotic equality
\begin{equation}\label{eq:largdevflat}
	\lim_{N\to \infty} \frac{1}{N} \ln \langle \fs | Q_{\delta \lambda } | \fs \rangle =  - \begin{cases} \frac{\delta^2}{2} , & \delta < \beta_c, \\
		\infty , & \delta \geq \beta_c  
		\end{cases}     
\end{equation}
for $\delta \lambda > \Gamma$. Unfortunately, \eqref{eq:largdevflat} follows from Theorem~\ref{thm:mainasym} only in the subregime $\delta \lambda > 2 \Gamma$ or $\delta \geq \beta_c$. The main obstacle for $ \Gamma < \delta\lambda \leq 2 \Gamma $ is that the moment generating function will be governed by the transversal field and thus yields no sharp information about $ \langle \fs | Q_{\delta \lambda } | \fs \rangle. $ We believe however that the methods used to prove Proposition~\ref{lem:qlargedev} can be refined to establish the conjectured identity \eqref{eq:largdevflat} for all $\lambda \delta > \Gamma$. As this requires considerable effort, a full proof is beyond the scope of this note.
\end{enumerate}

The proof of the quantum generalization~\eqref{eq:qlargedev} of this large-deviation principle  is the core of the novel technical result, and can be found in Subsection~\ref{subsec:PP}. It starts by 
rewriting eigenprojection $ Q_{\delta \lambda } $ in terms of the corresponding normalized eigenfunctions $ \psi(E) $ with energies $ E > \delta \lambda N   $,
	\begin{equation}\label{eq:efexp} 
		\langle \fs | Q_{\delta \lambda } | \fs \rangle = \sum_{E > \delta \lambda N } \left| \langle \fs |  \psi(E)\rangle \right|^2 . 
	\end{equation} 
	As we will establish through a representation, these eigenfunctions are in correspondence to the set $ \mathcal{L}_\delta  $. We will also prove  $ \ell^1 $-properties of these functions.\\

\subsection{Proof of Theorem~\ref{thm:mainasym}}
Taking Proposition~\ref{lem:qlargedev} for granted, the proof of the main results proceeds as follows.
\begin{proof}[Proof of Theorem~\ref{thm:mainasym}] 
For a lower bound, we employ Lemma~\ref{lem:lb} to the case of the REM. Its assumption~\eqref{eq:lemlb}  is satisfied for $\mathbb{P} $-almost all realizations of the REM by the strong law of large numbers. By the known almost sure convergence~\eqref{eq:REMp} of the REM's pressure, one has $ p_U = p_{\textrm{REM}} $.  
This yields~\eqref{eq:mainasym} as a lower bound.

A complementing upper bound is based on Proposition~\ref{lem:qlargedev}. We set $ \Gamma = e^{-s} $ and pick $ \varepsilon > 0 $ arbitrary to estimate:
\begin{align}\label{eq:upperb}
e^{tN}  \langle \fs | e^{t W_{\lambda,s} } | \fs \rangle & =   \langle \fs | e^{t ( \Gamma T-\lambda U)  } (1 - Q_{\Gamma + \varepsilon}) | \fs \rangle  +   \langle \fs |   e^{t ( \Gamma T-\lambda U)  } Q_{\Gamma + \varepsilon}  | \fs \rangle \notag \\
& \leq e^{t N (\Gamma + \varepsilon)}  \langle \fs |  (1 - Q_{\Gamma + \varepsilon})  | \fs \rangle  - \int_{ (\Gamma + \varepsilon) N}^\infty \mkern-20mu  e^{tE} \langle \fs |  d Q_{E/N}  | \fs \rangle  \notag \\
& \leq e^{t N (\Gamma + \varepsilon)}  + Nt \lambda \int_{ (\Gamma + \varepsilon)/ \lambda}^\infty \mkern-20mu e^{Nt\lambda x}  \langle \fs |   Q_{x \lambda }  | \fs \rangle  \ dx . 
\end{align}
In the second step, the second term was rewritten as an integral with respect to the spectral measure of $  \Gamma T-\lambda U $ and the flat vector. 
The third step results from an integration by parts, followed by a change of variables, and combining the second term with the first.
Based on Proposition~\ref{lem:qlargedev}, an asymptotic evaluation of the last integral yields the upper bound: 
$$
\limsup_{N\to \infty} \frac{1}{N t } \ln \int_{ (\Gamma + \varepsilon)/ \lambda}^\infty \mkern-20mu e^{Nt\lambda x}  \langle \fs |   Q_{x \lambda }  | \fs \rangle  \ dx \leq  t^{-1} p_{\textrm{REM}}(t \lambda ) , 
$$
which, in competition with the first term in~\eqref{eq:upperb}, implies 
$$
\limsup_{N\to \infty} \frac{1}{N t } \ln e^{tN}  \langle \fs | e^{t W_\lambda} | \fs \rangle \leq \max\left\{ \Gamma + \varepsilon   , t^{-1} p_{\mathrm{REM}}(t\lambda)  \right\}. 
$$
Since $ \varepsilon > 0 $ is arbitrary, this concludes the proof.
\end{proof}

\section{Properties of the QREM}

The remainder of this paper concerns the analysis of properties of the Quantum Random Energy Model and, in particular, the proof of the large deviation result, Proposition~\ref{lem:qlargedev}. However, we start with a topic of independent interest, namely $  \ell^p $-properties of the resolvent of more general, deterministic potentials on the Hamming cube. 

\subsection{Deterministic $ \ell^p $-analysis of the resolvent}\label{sec:L1}
In this subsection, we look at Hamiltonians of the form
$$
H := \Gamma T - U \quad \mbox{on $ \mathcal{H} \equiv \ell^2(\mathcal{Q}_N)  $ }
$$

with $ \Gamma \geq  0 $ and a multiplication operator corresponding to a fixed deterministic potential $ U : \mathcal{Q}_N \to \mathbb{R} $ (not necessarily the REM). Recall that the flat state~\eqref{eq:flat} is the $ \ell^2 $-normalized eigenvector of $ T $ corresponding to its maximal eigenvalue $ N $. Subsequently, we assume 
\begin{equation}\label{eq:resset}
E > \max\{ \max \mathrm{spec}(H) , \, \Gamma N , \,-  \inf U \} ,
\end{equation} 
and investigate properties of the resolvent operator $ R(E) := (E - H )^{-1} $. The following summarizes some basic facts in the regime~\eqref{eq:resset}. 
\begin{enumerate}
	\item The resolvent's matrix elements are all non-negative, $  \langle \pmb{\tau} | R(E) |  \pmb{\sigma} \rangle \geq 0 $ for any $  \pmb{\sigma},  \pmb{\tau} \in \mathcal{Q}_N $. This follows from the identity
    $$ \langle \pmb{\tau} | R(E) |  \pmb{\sigma} \rangle = \int_0^\infty e^{-tE} \langle \pmb{\tau} | e^{tH}  |  \pmb{\sigma} \rangle ds $$ 
    relating the resolvent to the semigroup and the Feynman-Kac representation for the latter, and the fact
    that $ T $ has non-negative matrix elements. 
	\item By the above relation to the semigroup, the resolvent's matrix elements are  monotone decreasing in the potential, i.e., $ U(\pmb{\sigma})  \leq U'(\pmb{\sigma}) $ for all $  \pmb{\sigma} \in \mathcal{Q}_N $ implies  
	\begin{equation}
	\langle \pmb{\tau} | R(E) |  \pmb{\sigma} \rangle \geq  \langle \pmb{\tau} | R'(E) |  \pmb{\sigma} \rangle 
	\end{equation}
	 for any $  \pmb{\sigma},  \pmb{\tau} \in \mathcal{Q}_N $. 
\end{enumerate}
Our main aim is to control the properties of the resolvent operator, when this operator acts on the $ \ell^p $-space, that is, $ R(E) f (\pmb{\sigma}) = \sum_{\pmb{\tau}}  \langle \pmb{\sigma} | R(E) |  \pmb{\tau} \rangle f(  \pmb{\tau} ) $ for $ f \in \ell^p(\mathcal{Q}_N) $. The main difference is the norm $ \| f \|_p := (  \sum_{\pmb{\tau}}  | f( \pmb{\tau}) |^p )^{1/p} $. Hence, 
\begin{equation}
L_p(E) := \sup_{ \| f\|_p = 1} \left\| R(E) f \right\|_p 
\end{equation}
stands for the operator norm of the resolvent on $ \ell^p(\mathcal{Q}_N) $. For $ p = 2, $ one has the spectral estimate $ L_2(E) \leq ( E -  \max \mathrm{spec}(H) ) ^{-1} $. By self-adjointness, one has the duality 
$ L_p(E) = L_q(E) $ for H\"older conjugated exponents $ p^{-1} + q^{-1} = 1 $ with $p,q \in [1,\infty]$. Indeed, the Hölder duality implies 
\begin{align*}
 L_p(E) &= \sup_{ \| f\|_p = 1} \left\| R(E) f \right\|_p = \sup_{ \| f\|_p = 1, \, \|g\|_{q} = 1} \left| \langle g, R(E) f \rangle \right| = \sup_{ \| f\|_p = 1, \, \|g\|_{q} = 1} \left| \langle R(E)^*g, f \rangle \right| \\&= \sup_{ \| f\|_p = 1, \, \|g\|_{q} = 1} \left| \langle R(E)g, f \rangle \right| = \sup_{ \| g\|_q = 1} \left\| R(E) f \right\|_q = L_q(E),  
 \end{align*}
which is even valid for $p=1$ and $q = \infty$ as we are considering finite-dimensional vector spaces. To cover the full range of $ p\in [1,\infty] $ with the help of interpolation, it therefore remains to investigate case $ p = 1 $. Due to the non-negativity of the resolvent's matrix elements, we have 
\begin{equation}
L_1(E) = \sup_{\pmb{\sigma} \in Q_N} \sum_{\pmb{\tau} \in  \mathcal{Q}_N}  \langle \pmb{\tau} | R(E) |  \pmb{\sigma} \rangle  = \sup_{\pmb{\sigma} \in Q_N} \sqrt{2^N} \ \langle \fs | R(E) |  \pmb{\sigma} \rangle 
\end{equation}
in terms of the flat state $|\fs \rangle$. An upper bound for $L_1(E)$ is achieved in the following. 
\begin{proposition}
Suppose~\eqref{eq:resset} and that
\begin{equation}
\gamma_N(E) :=  \sup_{\pmb{\sigma} \in  \mathcal{Q}_N}  \frac{1}{ N } \sum_{d(\pmb{\tau}, \pmb{\sigma} ) = 1}  \frac{U_-(\pmb{\tau}) } { E- U_-(\pmb{\tau})  } < \frac{E - \Gamma N }{\Gamma N} ,
\end{equation} 
where $ U_-(\pmb{\tau}) := \max \{ - U(\pmb{\tau}) , 0 \} $ denote the negative part. 
Then, for any   $ \pmb{\sigma} \in  \mathcal{Q}_N $: 
\begin{equation}\label{eq:L1est}
\sqrt{2^N} \ \langle \fs | R(E) |  \pmb{\sigma} \rangle \leq  \frac{1}{E-\Gamma N (1+\gamma_N(E) )} \frac{E } { E+ U(\pmb{\sigma}) } .
\end{equation}
\end{proposition}
\begin{proof}
The proof is a simple consequence of a twofold application of the resolvent equation and the fact that $ T | \fs \rangle = N | \fs \rangle $:
\begin{align}
	\sqrt{2^N} \ \langle \fs | R(E) |  \pmb{\sigma} \rangle   & = \sqrt{2^N}\left[ \langle \fs | \frac{1}{E-\Gamma T} |  \pmb{\sigma} \rangle - \langle \fs | \frac{1}{E-\Gamma T} U R(E) |  \pmb{\sigma} \rangle \right] \notag \\
    &=\frac{1}{E-\Gamma N} \left[  1 - \sqrt{2^N} \langle \fs | U R(E) |  \pmb{\sigma} \rangle  \right] \notag \\
    &= \frac{1}{E-\Gamma N} \left[  1 - \sqrt{2^N} \langle \fs | U \left( \frac{1}{E+U} + \frac{1}{E+U} \Gamma T R(E)\right) |  \pmb{\sigma} \rangle  \right] \notag \\
	& =  \frac{1}{E-\Gamma N} \left[  1 - \frac{U(\pmb{\sigma}) } { E+U(\pmb{\sigma})  } - \sqrt{2^N} \langle \fs | \frac{U}{E + U} \Gamma T R(E)  |  \pmb{\sigma} \rangle \right] \notag \\
	& =  \frac{1}{E-\Gamma N} \left[   \frac{E } { E+U(\pmb{\sigma}) } - \Gamma\mkern-10mu \sum_{\substack{\pmb{\tau},\pmb{\tau}' \in \mathcal{Q}_N \\ d(\pmb{\tau}, \pmb{\tau}' ) = 1}}    \frac{U(\pmb{\tau}) } { E+U(\pmb{\tau}) }\    \langle \pmb{\tau}' | R(E) |  \pmb{\sigma} \rangle \right] .
\end{align}
Since the matrix elements of the resolvent are non-negative, the last term on the right is bounded from above by
$$
 \frac{\Gamma}{ E-\Gamma N } \left(\sup_{\pmb{\sigma}' \in  \mathcal{Q}_N} \sum_{d(\pmb{\tau}, \pmb{\sigma}' ) = 1}  \frac{U_-(\pmb{\tau}) } {E - U_-(\pmb{\tau}) }  \right) \sum_{\pmb{\tau}' \in  \mathcal{Q}_N}  \langle \pmb{\tau}' | R(E) |  \pmb{\sigma} \rangle = \gamma_N(E)  \frac{\Gamma N }{ E-\Gamma N }\  \sqrt{2^N} \langle \fs | R(E) |  \pmb{\sigma} \rangle . 
$$ 
This term may thus be subtracted from the left side. Since $ \gamma_N(E)  \frac{\Gamma N }{ E-\Gamma N } < 1 $, one arrives at the claim. 
\end{proof}
The above estimate~\eqref{eq:L1est} on the $ \ell^1 $-norm of the resolvent is only in spirit related to the classical $ L^p$-estimates for Schrödinger semigroups and resolvents~\cite{Sim82}. While the latter mainly cope with the singularities of potentials in continuous space, our estimate mainly chases the $ N $-dependence. However, they agree in that the behavior of the Laplacian is shown to prevail.\\

For an application in case $ U : \mathcal{Q}_N \to \mathbb{R} $ is a typical realization of the REM, we need the following probabilistic lemma.

\begin{lemma}
Let $ \varepsilon > 0 $. Furthermore, let $ \eta, \kappa > 0,  $ and $ E \geq  (\eta + \kappa) N $ and set 
$$ U_-(\pmb{\sigma}) = - U(\pmb{\sigma})  1[ - \eta N < U(\pmb{\sigma}) \leq 0 ] ,
$$ 
where $ U $ is the REM. 
Then there is a constant $ c = c(\varepsilon, \eta, \kappa) $ such that 
for all $ N $ large enough
\begin{equation}\label{eq:REMprob}
	\mathbb{P}\left(  \sup_{\pmb{\sigma} \in  \mathcal{Q}_N}  \frac{1}{ N } \sum_{d(\pmb{\tau}, \pmb{\sigma} ) = 1}  \frac{U_-(\pmb{\tau}) } { E- U_-(\pmb{\tau})  } \leq \varepsilon \right) \geq 1 - e^{-c N } . 
\end{equation}
\end{lemma} 
\begin{proof}
We employ a union bound and estimate for fixed $ \pmb{\sigma} $ the probability of
\begin{align}
 \varepsilon N < \sum_{d(\pmb{\tau}, \pmb{\sigma} ) = 1}  \frac{U_-(\pmb{\tau}) } { E- U_-(\pmb{\tau})  } \leq & \ \frac{\alpha N^{2}}{E - \alpha N} +  \frac{\eta  N  }{E - \eta N}\ N_\alpha(\pmb{\sigma})  , \notag \\
& \mbox{with} \quad N_\alpha(\sigma) :=  \sum_{d(\pmb{\tau}, \pmb{\sigma} ) = 1}  1[  U(\pmb{\sigma}) \leq -\alpha N  ] , \quad \alpha > 0 .
\end{align}
Since 
$$
\mathbb{P}\left( U(\pmb{\sigma}) \leq -\alpha N \right) = \int_{-\infty}^{-\alpha \sqrt{N} } \mkern-10mu \exp\left( -\frac{u^2}{2}\right) \frac{du}{\sqrt{2\pi}} \leq \exp\left(- \frac{N \alpha^2}{2}\right)  =: p_{\alpha,N} ,
$$
a standard concentration of measure estimate for the Bernoulli counting variable $ N_\alpha(\pmb{\sigma})  $  shows that for any $ \delta > 0 $:
$$
\mathbb{P}\left( N_\alpha(\pmb{\sigma})  > \delta N \right) \leq \exp\left( - N  \delta \ln\left( \frac{\delta}{p_{\alpha,N}}\right) + N (1-\delta) \ln\left(\frac{1 - \delta}{1-  p_{\alpha,N}}\right) \right) . 
$$
The exponent on the right side is quadratically bounded, i.e., by $- N^2 c(\alpha,\delta) $ with some $ c(\alpha,\delta) > 0 $. Choosing $\alpha, \delta > 0$ small enough such that $ \frac{\alpha}{\eta - \alpha} + \frac{\eta \delta}{\kappa} < \varepsilon$ completes the proof. 
\end{proof}
\subsection{Eigenfunctions above the localization threshold} 
This subsection is devoted to properties of eigenfunctions $ \psi(E) \in \mathcal{H}  $ of the Quantum Random Energy Model (QREM) $$ H := \Gamma T - \lambda U $$ for energies $ E > \lambda \delta N > \Gamma N $. This includes the regime of localization close to the maximum of the spectrum near  $ \beta_c N$ covered in~\cite{MaWa22}, but substantially goes beyond it. In fact, although we do not quite prove localization~\cite{AizWar15}, we cover what is believed~\cite{LPS14,BSc25}  to be the full localization regime up to the optimal threshold at $ \Gamma N$.  \\

The main tool for establishing properties of eigenfunctions with energies $ E > \lambda \delta N > \Gamma N $ is the Schur complement representation corresponding to the direct sum decomposition of the Hilbert space $ \mathcal{H} \equiv \ell^2(\mathcal{Q}_N)  $ into $ \ell^2(\mathcal{L}_\eta) \oplus  \ell^2(\mathcal{Q}_N \backslash \mathcal{L}_\eta) $ with the extreme deviations set $  \mathcal{L}_\eta $ from~\eqref{eq:ldset} and 
\begin{equation}\label{eq:eta}
\eta := \delta - \varepsilon  \quad \mbox{with} \quad 0 < \varepsilon < \delta - \Gamma / \lambda  
\end{equation}
arbitrary.  
Accordingly, let $ \mathbbm{1}^< $ and, respectively, $  \mathbbm{1}^> := \mathbbm{1} -  \mathbbm{1}^< $ stand for the orthogonal projection onto $ \ell^2(\mathcal{L}_\eta)  $ and its orthogonal complement, respectively. We denote by
$$
 H^> :=   \mathbbm{1}^>  (  \Gamma T - \lambda U  )    \mathbbm{1}^> , \quad \mbox{on $  \ell^2(\mathcal{Q}_N \backslash \mathcal{L}_\eta)  $}
$$
the canonical restriction of $ H $  to the subspace, on which $ U(\pmb{\sigma}) \geq - \eta N $. 
The largest eigenvalue of $  H^> $ is known to be bounded  from above by a constant strictly smaller than~$ \lambda \delta N $. 

\begin{lemma}[Prop. 4.5 in \cite{MaWa22}; see also Thm.~2.3 in \cite{MaWa20}]\label{lem:normH>}
For  $ \lambda \delta  > \Gamma \geq 0 $ and~\eqref{eq:eta}, there is some $ c = c(\delta,\eta,\varepsilon,\lambda,\Gamma) > 0 $ such that  for all $ N $ large enough
\begin{equation}\label{eq:normH>}
 \mathbb{P}\left( \max \mathrm{spec} (H^>)   \leq \lambda ( \delta - \varepsilon/2) N \right) \geq  1 - e^{-c N }  . 
\end{equation}
\end{lemma}
This implies that for any $ E > \lambda \delta N $ the operator $ E - H^> $ is boundedly invertible with inverse $  \left(E-  H^> \right)^{-1}  $ on $  \ell^2(\mathcal{Q}_N \backslash \mathcal{L}_\eta)  $. For any eigenvector $ \psi(E) $ of $ H $ with such an energy, the two components 
\begin{equation}
	|  \psi^<(E) \rangle \  := \mathbbm{1}^< | \psi(E) \rangle, \qquad 
	|  \psi^>(E) \rangle \  := \mathbbm{1}^> | \psi(E) \rangle
\end{equation}
are linked by the following relation
\begin{equation}
|  \psi^>(E) \rangle =  R^>(E) A |  \psi^<(E) \rangle , \qquad\mbox{with}\quad R^>(E) := \left(E-  H^> \right)^{-1} , \; A :=  \mathbbm{1}^> \Gamma T  \mathbbm{1}^< . 
\end{equation}
Consequently, we may express the scalar product of $ | \psi(E) \rangle $ with any vector in terms of a scalar product on the subspace $ \ell^2(\mathcal{L}_\eta)  $ only. 
In~\eqref{eq:efexp}, the scalar product with the flat state is key:
\begin{equation*}
	\langle \fs |  \psi(E) \rangle = \frac{1}{\sqrt{2^N} } \sum_{\pmb{\sigma}\in \mathcal{L}_\eta} \langle \pmb{\sigma} | \psi^<(E) \rangle \left[ 1 + \sqrt{2^N} \langle \fs |  \mathbbm{1}^>  R^>(E) A |  \pmb{\sigma} \rangle \right]  =  \frac{1}{\sqrt{2^N} } \langle \varphi(E) | \psi^<(E) \rangle ,
\end{equation*}
with 
\begin{equation}\label{eq:defvarphi}
 | \varphi(E) \rangle := \sqrt{2^N}  \left( \mathbbm{1}^< + A^* R^>(E)  \mathbbm{1}^>\right)  | \fs \rangle \in  \ell^2(\mathcal{L}_\eta)  . 
\end{equation}
The main observation is that in the parameter regime of interest and with asymptotically full probability, this vector is in $ \ell^\infty(\mathcal{Q}_N)  $ with a norm bounded uniformly in $ N $. Moreover, it is a smooth function of $ E $, since the resolvent is known to be an analytic function in the resolvent set. Our smooth-family argument  will also be based on $ \ell^\infty $-bounds on the $ k $th derivative of the above vector:
\begin{equation}
 | \varphi^{(k)}(E) \rangle := (-1)^k  \sqrt{2^N}  A^*  R^>(E)^{k+1}  \mathbbm{1}^>  | \fs \rangle \in  \ell^2(\mathcal{L}_\eta)  , \quad k \in \mathbb{N} .
\end{equation}
The key $ \ell^\infty $-estimates are the following. 
\begin{lemma}\label{lem:phiest}
For  $ \lambda \delta  > \Gamma \geq 0 $ and assuming~\eqref{eq:eta}, there is some $ c \equiv c(\lambda\delta,\eta,\Gamma) >0 $ and an event of probability at least $ 1 - e^{-c N } $ such that on this event and for all $ N $, all $ E \geq \delta \lambda N $ and all $ \pmb{\sigma}\in \mathcal{L}_\eta $:
\begin{equation}\label{eq:phiLinfty}
	 \left| \langle \pmb{\sigma} | \varphi(E) \rangle \right|  \leq 1 + \frac{2\Gamma}{\lambda\delta -\Gamma} \ \frac{\delta}{\varepsilon}   ,
\end{equation}
and for any $ k \in \mathbb{N} $:
\begin{equation}\label{eq:phiLinftyk}
 \left| \langle \pmb{\sigma} | \varphi^{(k)}(E) \rangle \right|  \leq  \frac{\Gamma}{N^k} \ \left(\frac{2}{\lambda\delta -\Gamma} \ \frac{\delta}{\varepsilon}\right)^{k+1} . 
\end{equation}
\end{lemma}
\begin{proof}
Throughout the proof, we restrict to the event on which~\eqref{eq:normH>} holds, which by Lemma~\ref{lem:normH>} has a probability that is exponentially close to one.
In order to be able to use the results of the previous Subsection~\ref{sec:L1}, we further note that the restricted resolvent can be viewed as the limit $ \alpha \to \infty $ of the resolvent $ (E+\lambda U_\alpha - \Gamma T)^{-1} $ on $ \ell^2(\mathcal{Q}_N) $, in which the value of $ U_\alpha: \mathcal{Q}_N \to \mathbb{R} $ on $ \mathcal{L}_\eta $ is set to $ \alpha > 0 $. 
Therefore, the estimate~\eqref{eq:L1est} carries over to the restricted resolvent. Using the restriction $ E \geq \lambda \delta N > \Gamma N $  together with~\eqref{eq:eta}, this implies that for all $   \pmb{\sigma} \in \mathcal{L}_\eta $:
\begin{equation}
0 \leq \sqrt{2^N} \ \langle \fs | \mathbbm{1}^> R^>(E) |  \pmb{\sigma} \rangle \leq  \frac{1}{\lambda\delta N -\Gamma N (1+\gamma_N(E) )} \ \frac{\delta}{\varepsilon}   
\end{equation}
with 
$$ \gamma_N(E) =  \sup_{\pmb{\sigma} }  \frac{1}{ N } \sum_{d(\pmb{\tau}, \pmb{\sigma} ) = 1}  \frac{U_-(\pmb{\tau}) } { E- U_-(\pmb{\tau})  } \leq \frac{\lambda \delta - \Gamma }{2\Gamma} . 
$$
The last inequality holds on the event in~\eqref{eq:REMprob} (with $ \varepsilon = \frac{\lambda \delta - \Gamma }{2\Gamma}  $ there), which has the desired probability. 
Consequently, on this event and for all $   \pmb{\sigma} \in \mathcal{L}_\eta $ and all $ N $:
$$
  \left| \langle \pmb{\sigma} | \varphi(E) \rangle \right|  \leq 1 + \Gamma\sum_{ \substack{\pmb{\tau}  \in \mathcal{Q}_N\backslash \mathcal{L}_\eta \\ d(\pmb{\sigma},\pmb{\tau}) = 1} } \sqrt{2^N} \ \langle \fs | \mathbbm{1}^> R^>(E) |  \pmb{\tau} \rangle \leq 1 + \frac{2\Gamma}{\lambda\delta -\Gamma} \ \frac{\delta}{\varepsilon}   . 
$$
This completes the proof of~\eqref{eq:phiLinfty}. For the proof of~\eqref{eq:phiLinftyk}, we iterate the above estimate 
\begin{align}
&  \left| \langle \pmb{\sigma} | \varphi^{(k)}(E) \rangle \right|   \leq \Gamma \sum_{ \substack{\pmb{\tau} \in \mathcal{Q}_N\backslash \mathcal{L}_\eta \\ d(\pmb{\sigma},\pmb{\tau}) = 1} } \sqrt{2^N} \ \langle \fs | \mathbbm{1}^> R^>(E)^{k+1} |  \pmb{\tau} \rangle \notag \\
 & \leq N\Gamma \mkern-15mu \sup_{ \pmb{\tau}_1, \dots \pmb{\tau}_{k+1} \in  \mathcal{Q}_N\backslash \mathcal{L}_\eta} \prod_{j=1}^{k+1} \ \sqrt{2^N} \langle \fs | \mathbbm{1}^> R^>(E) |  \pmb{\tau}_j \rangle \leq \Gamma N^{-k} \left(\frac{2}{\lambda\delta -\Gamma} \ \frac{\delta}{\varepsilon}\right)^{k+1} .
\end{align}
The second estimate results from representing the $ k+1 $-fold product as a $ k $-fold summation, and estimating those sums over the non-negative matrix elements of the resolvent using a trivial $ p = 1 $ and $ q = \infty $ H\"older bound. 
\end{proof}
\subsection{Proof of Proposition~\ref{lem:qlargedev}}\label{subsec:PP}

\begin{proof}[Proof of Proposition~\ref{lem:qlargedev}]
We  pick $ \varepsilon > 0 $  small satisfying~\eqref{eq:eta}, but otherwise arbitrary, and set $ \eta = \delta - \varepsilon $. 
		The proof starts from the the expansion~\eqref{eq:efexp}  in terms of the corresponding normalized eigenfunctions $ \psi(E) $ with energies $ E > \delta \lambda N  > \Gamma N $. Using~\eqref{eq:defvarphi} we rewrite the scalar product  in terms of the vectors $   \psi^<(E),  \varphi(E) \in \ell^2(\mathcal{L}_\eta ) $ on the subspace corresponding to extreme deviations:	
		$$ \langle \fs | Q_{\delta \lambda } | \fs \rangle = \frac{1}{2^N} \sum_{E \geq \delta \lambda N }  \left| \langle \varphi(E) |  \psi^<(E) \rangle \right|^2 . 
	$$
	Note that there is an event with probability exponentially close to unity, on which the summation in~\eqref{eq:efexp} can be restricted to the bounded interval 
	$ I_{\delta,\lambda}  := (\delta \lambda N , 2\beta_c N ) $. We decompose this interval into $ N_{\varepsilon} \in \mathbb{N} $ subintervals $ I_\varepsilon(n) $ of length bounded by $ e^{-\varepsilon N } $:
\begin{equation}\label{eq:intervsum}
		\langle \fs | Q_{\delta \lambda } | \fs \rangle  = \frac{1}{2^N} \sum_{n=1}^{N_\varepsilon} \sum_{E \in I_\varepsilon(n)} \left| \langle \varphi(E)  |  \psi^<(E)\rangle\right|^2 .
\end{equation}
For each interval $  I_\varepsilon(n) $, we pick the center $ E_n $ of this interval. By Taylor's theorem, for any $ E \in I_\varepsilon(n) $ and $ k \in \mathbb{N} $, there is $ \widehat E_n^{(k)} \in I_\varepsilon(n) $ such that 
\begin{align}
	\langle \varphi(E)  |  \psi^<(E)\rangle = & \sum_{j=0}^{k-1} \langle \varphi^{(j)}(E_n)  |  \psi^<(E)\rangle \frac{(E-E_n)^j}{j!} + r_n^{(k)}(E) \\
	& \mbox{with} \quad r_n^{(k)}(E) := \langle \varphi^{(k)}(\widehat E_n^{(k)})  |  \psi^<(E)\rangle \frac{(E-E_n)^k}{k!} . \notag 
\end{align}
We now further restrict attention to the event in Lemma~\ref{lem:phiest}, which still has a probability exponentially close to one. On this event, by~\eqref{eq:phiLinftyk} and the Cauchy-Schwarz estimate 
\begin{align}
\left| \langle \varphi^{(k)}( E')  |  \psi^<(E)\rangle  \right| & \leq \big\|   \varphi^{(k)}(\widehat E) \big\|_2 \big\| \psi^<(E) \|_2 \notag \\
&  \leq  \sup_{\pmb{\sigma} \in \mathcal{L}_\eta} \sup_{E' > \lambda \delta N} \left| \langle \pmb{\sigma} |  \varphi^{(k)}(E') \rangle \right| \sqrt{ \left| \mathcal{L}_\eta \right| }  , 
\end{align}
valid for any  for any $ E, E' > \lambda \delta N $, we obtain the bound 
\begin{equation}
\sup_{E \in I_\varepsilon(n)} \left| r_n^{(k)}(E)  \right| \leq c \ \frac{C^{k+1} e^{-\varepsilon k N}}{2^k N^k k!} \sqrt{ \left| \mathcal{L}_\eta \right| } , \quad \mbox{with}\quad  C := \frac{2}{\lambda\delta -\Gamma} \ \frac{\delta}{\varepsilon} , \; c := \Gamma + C^{-1} .
\end{equation}
The choice $ k = k_N:= e^{\varepsilon N} $ ensures that the last term is bounded by $ \sqrt{2^{-N}  \left| \mathcal{L}_\eta \right|} $ for all $ N $ large enough. 

The sum in~\eqref{eq:intervsum} over the eigenvalues for a fixed interval may be split into two parts:
\begin{align}\label{eq:twosums}
	 \sum_{E \in I_\varepsilon(n)} \left| \langle \varphi(E)  |  \psi^<(E)\rangle\right|^2 \leq & \ 2 k_N \sum_{j=0}^{k_N-1}  \frac{e^{-\varepsilon jN}}{j!}  \sum_{E \in I_\varepsilon(n)} \left|  \langle \varphi^{(j)}(E_n)  |  \psi^<(E)\rangle \right|^2 + 2  \sum_{E \in I_\varepsilon(n)}  \left| r_n^{(k_N)}(E)  \right|^2 \notag \\
	 =: & \ S^{(1)}_N(n) + S^{(2)}_N(n) . 
\end{align}
Inserting the above estimate into the second contribution leads to the bound
\begin{equation}\label{eq:2nd}
\frac{1}{2^N} \sum_{n=1}^{N_\varepsilon} S^{(2)}_N(n) \leq \frac{2}{2^{N}} \sum_{n=1}^{N_\varepsilon} \sum_{E \in I_\varepsilon(n)} 2^{-N}  \left| \mathcal{L}_\eta \right| \leq 2^{1-N}  \left| \mathcal{L}_\eta \right| . 
\end{equation}
The first sum in~\eqref{eq:twosums} is further estimated using the fact that $ | \psi^<(E) \rangle = \mathbbm{1}^<  | \psi(E) \rangle $ and the orthonormality of the eigenbasis $ \psi(E) $, which implies that for all $ j \in \mathbb{N} $:
$$
 \sum_{E \in I_\varepsilon(n)} \left|  \langle \varphi^{(j)}(E_n)  |  \psi^<(E)\rangle \right|^2  \leq \big\|  \varphi^{(j)}(E_n) \big\|_2^2 \leq  \left| \mathcal{L}_\eta \right|  \sup_{E \geq \lambda\delta N} \sup_{\pmb{\sigma} } \left| \langle \pmb{\sigma} | \varphi^{(j)}(E) \rangle \right|^2 \leq   \left| \mathcal{L}_\eta ´\right|   \Gamma^2 \frac{C^{2j+2}}{N^{2j}} . 
$$
Consequently, the contribution of the first sum to the total sum in~\eqref{eq:twosums} is bounded by
\begin{equation}\label{eq:1st}
\frac{1}{2^N} \sum_{n=1}^{N_\varepsilon} S^{(1)}_N(n)  \leq 2 k_N \Gamma^2C^2   \exp\left( \frac{C^2}{N^2} e^{-\varepsilon N} \right) \frac{\left| \mathcal{L}_\eta \right|}{2^N}   \leq 2 \Gamma^2 C^2 e^{C^2} e^{\varepsilon N}  \ \frac{\left| \mathcal{L}_\eta \right|}{2^N} . 
\end{equation}
  In view of the  classical large deviation result in \eqref{eq:clLD},
  our bounds on the second and first contributions, \eqref{eq:2nd} and~\eqref{eq:1st} respectively, imply that on an event of probability exponentially close to one:  
$$
\limsup_{N\to \infty} \frac{1}{N} \ln \langle \fs | Q_{\delta \lambda } | \fs \rangle  \leq \varepsilon - 
\begin{cases} \frac{(\delta-\varepsilon)^2}{2} , & \delta -\varepsilon < \beta_c, \\
		\infty , & \delta -\varepsilon \geq \beta_c .
		\end{cases}
$$
The bound~\eqref{eq:qlargedev} then follows from a Borel-Cantelli argument and the fact that $ \varepsilon > 0 $ can be chosen arbitrarily small. 
\end{proof}

\appendix 

\section{Supplementary Results}

This appendix is mainly concerned with the supplementary results concerning the trace of the QREM's spectral projector in \eqref{eq:traceLD}  and the sharpness of the estimate in Proposition~\ref{lem:qlargedev}.
For the proof of $\eqref{eq:traceLD}$, we recall that the eigenvalues $E > \Gamma N$ of the QREM Hamiltonian are given by small shifts of the classical low-energy  levels. The following lemma gives a crude estimate, which is enough for our purposes in case $\delta < \beta_c$.

\begin{lemma}\label{lem:shift} Let $E_1 \geq E_2 \geq \cdots \geq E_{2^N} $ be the order eigenvalues of the QREM Hamiltonian $H=\Gamma T -\lambda U $  and label the classical configurations $\pmb{\sigma}$ in the opposite order, i.e. $U(\pmb{\sigma}_1) \leq U(\pmb{\sigma}_2) \leq \cdots \leq U(\pmb{\sigma}_{2^{N}})$. Suppose that $\delta \lambda > \Gamma$. Then, there are some constants $K = K(\lambda \delta, \Gamma) \in \mathbb{N}$ and $c = c(\lambda \delta, \Gamma) > 0$ such that 
\begin{equation}\label{eq:shift}
   \mathbb{P}\left(\sup_{j: E_j > \lambda \delta N} |E_j + \lambda U(\pmb{\sigma}_j)| > K \sqrt{N}\right) \leq e^{-cN}
\end{equation}
\end{lemma}

Lemma~\ref{lem:shift} is a consequence of the truncation method - isolating the deep holes of REM -first introduced in \cite{Manai:2020ta}. To be more precise,  Lemma~2 and Lemma~3 and the decomposition~(15) in \cite{Manai:2020ta} imply Lemma~\ref{lem:shift}. The next lemma addresses the claimed trace of the QREM's spectral projector in \eqref{eq:traceLD}. 
\begin{lemma}\label{lem:TraceProj}
    Let $U$ be the REM  and $Q_{\delta \lambda}$ the spectral projector defined in \eqref{eq:proj}. Then, for any $\lambda \delta > \Gamma$, and almost all realizations of the REM
    \begin{equation}\label{eq:traceLDA}
	\lim_{N\to \infty} \frac{1}{N} \ln 2^{-N } \tr Q_{\delta \lambda }   =  - \begin{cases} \frac{\delta^2}{2} , & \delta < \beta_c, \\
		\infty , & \delta \geq \beta_c . 
		\end{cases}     
\end{equation}
\end{lemma}
\begin{proof}
    We distinguish between the cases $\delta < \beta_c$ and $\delta \geq \beta_c$. Suppose first that $\delta < \beta_c$ and let $\varepsilon > 0$ be small enough such that $\Gamma < \lambda (\delta - \varepsilon)$ and $\delta + \varepsilon < \beta_c$. Except for an event of exponentially small probability, Lemma~\ref{lem:shift} implies for $N$ large enough
    $$ |\mathcal{L}_{\delta + \varepsilon}| \leq \tr Q_{\delta \lambda } \leq |\mathcal{L}_{\delta - \varepsilon}|. $$
    Due to the classical large deviation result \eqref{eq:clLD}, a Borel-Cantelli argument yields the almost sure bounds 
    $$ -\frac{(\delta + \varepsilon)^2}{2} \leq \liminf_{N\to \infty} \frac{1}{N} \ln 2^{-N } \tr Q_{\delta \lambda } \leq \limsup_{N\to \infty} \frac{1}{N} \ln 2^{-N } \tr Q_{\delta \lambda } \leq -\frac{(\delta - \varepsilon)^2}{2}. $$
    Since $\varepsilon > 0$ can be chosen arbitrarily small, the assertion follows for $\delta < \beta_c$.

It remains the case $\delta \geq \beta_c$. Here, we need  the ground state analysis from \cite[Theorem~1.5]{MaWa22} which guarantees that
\[  \max \mathrm{spec}(\Gamma T - \lambda U) = -\lambda \min_{\pmb{\sigma}} U(\pmb{\sigma}) + \frac{\Gamma^2}{\lambda \beta_c} +o(1), \]
holds almost surely if $\lambda \beta_c > \Gamma.$ Furthermore, we refer to the well-known asymptotics~\eqref{eq:REMmin} of the REM's minimum energy.
These two results imply that almost surely
\begin{equation}\label{eq:tracevanish}
\lim_{N \to \infty} \tr Q_{\delta \lambda} = 0, 
\end{equation}
if $\beta_c \lambda > \Gamma$ and $\delta \geq \beta_c $. This completes the proof.
    
\end{proof}

 \section*{Acknowledgments}
 SW thanks the DFG for support under grant EXC-2111 -- 390814868. 
 
 \noindent

\end{document}